\documentclass[10pt]{article}
\usepackage{}

\usepackage{pgf}
\usepackage[square, comma, sort&compress, numbers]{natbib}
\usepackage[english]{babel}
\usepackage[utf8]{inputenc}
\usepackage{datetime}
\usepackage{geometry}
\usepackage{array}
\usepackage{xspace}

\usepackage{amsmath,amsfonts,amssymb,amsopn,amsthm,amscd}
\theoremstyle{plain}
\usepackage{graphicx} 
\usepackage{epstopdf}
\usepackage{enumerate}
\allowdisplaybreaks[4]

\def\Box{\vcenter{\vbox{\hrule\hbox{\vrule
     \vbox to 8.8pt{\hbox to 10pt{}\vfill}\vrule}\hrule}}}

\usepackage{indentfirst}
\newtheorem{thm}{Theorem}[section]
\newtheorem{lem}[thm]{Lemma}

\newtheorem{cor}[thm]{Corollary}

\newtheorem{prop}[thm]{Proposition}

\newtheorem{exm}{Example}
\newtheorem{remark}{Remark}

\begin{document}

\begin{center}

{\large  \bf Minimal Linear Codes Constructed from Functions}

\vskip 0.8cm
{\small Xia Wu$^1$, Wei Lu$^1$\footnote{Supported by NSFC (Nos. 11971102, 11801070, 11771007)

  }}, Xiwang Cao$^2$\\

{\small $^1$School of Mathematics, Southeast University, Nanjing
210096, China}\\
{\small $^2$Department of Math, Nanjing University of Aeronautics and Astronautics, Nanjing 211100, China}\\
{\small E-mail:
 wuxia80@seu.edu.cn, luwei1010@seu.edu.cn, xwcao@nuaa.edu.cn}\\
{\small $^*$Corresponding author. (Email: luwei1010@seu.edu.cn)}
\vskip 0.8cm
\end{center}

{\bf Abstract:}
In this paper, we consider minimal linear codes by a general construction of linear codes from $q$-ary functions. First, we give a sufficient and necessary condition for
linear codes which are constructed by functions to be minimal. Second, as  applications, we present four constructions of minimal linear codes. Constructions on minimal linear codes  in this paper  generalize the results in some recent papers, and contain the results in these papers as special cases. In our four constructions, the choices of $f$ are much more flexible than theirs.


\section{\bf Introduction}
 Let $q$ be a prime power and $\mathbb{F}_q$ the finite field with $q$ elements. Let $n$ be a positive integer and $\mathbb{F}_q^n$  the vector space with dimension $n$ over $\mathbb{F}_q$. In this paper, all vector spaces are over $\mathbb{F}_q$ and  all vectors are row vectors.  For a vector $\mathbf{v}=(v_1, \dots, v_n)\in\mathbb{F}_q^n$, let Suppt$(\mathbf{v})$ $:= \{1 \leq i\leq n : v_i\neq 0\}$ be the support of $\mathbf{v}$. For any two vectors $\mathbf{u}, \mathbf{v}\in \mathbb{F}_q^n$, if $\rm{Suppt}(\mathbf{u})\subseteq \rm{Suppt}(\mathbf{v})$, we say that $\mathbf{v}$ covers $\mathbf{u}$ (or $\mathbf{u}$ is covered by $\mathbf{v}$) and write $\mathbf{u}\preceq\mathbf{v}$. Clearly, $a\mathbf{v}\preceq \mathbf{v}$ for all $a\in \mathbb{F}_q$.

 An $[n,k]_q$ linear code $\mathcal{C}$ over $\mathbb{F}_q$ is a $k$-dimensional subspace of $\mathbb{F}_q^n$. Vectors in $\mathcal{C}$ are called codewords. A codeword $\mathbf{c}$ in a linear code $\mathcal{C}$ is called \emph{minimal} if $\mathbf{c}$ covers only the codewords $a\mathbf{c}$ for all $a\in \mathbb{F}_q$, but no other codewords in $\mathcal{C}$. That is to say, if a codeword $\mathbf{c}$  is minimal in  $\mathcal{C}$, then for any codeword $\mathbf{b}$ in $\mathcal{C}$, $\mathbf{b}\preceq \mathbf{c}$ implies that $\mathbf{b}=a\mathbf{c}$ for some $a\in \mathbb{F}_q$.
  For an arbitrary linear code $\mathcal{C}$, it is  hard  to determine the set of its minimal codewords, \cite{BMT1978, BN1990}.


%
%

If every codeword in  $\mathcal{C}$ is minimal, then $\mathcal{C}$ is said to be a \emph{minimal linear code}. Minimal linear codes have interesting applications in secret sharing \cite{CDY2005, CCP2014, DY2003, M1995, YD2006} and secure two-party computation \cite{ABCH1995, CMP2013}, and could be decoded with a minimum distance decoding method \cite{AB1998}. Searching for minimal linear codes has been an interesting research topic in coding theory and cryptography.

The \emph{Hamming weight} of a vector $\mathbf{v}$ is wt$(\mathbf{v}):=\#\rm{Suppt}(\mathbf{v})$. In \cite{AB1998}, Ashikhmin and Barg gave a sufficient condition  on the minimum and maximum nonzero Hamming weights for a linear code to be minimal:
\begin{lem}\label{Ashikhmin-Barg}{$($\rm {Ashikhmin-Barg} \cite{AB1998})}
A linear code $\mathcal{C}$ over $\mathbb{F}_q$ is minimal if
$$\frac{w_{\rm min}}{w_{\rm max}}>\frac{q-1}{q},$$
where $w_{\rm min}$ and $w_{\rm max}$ denote the minimum and maximum nonzero Hamming weights in the code $\mathcal{C}$, respectively.
\end{lem}

Inspired by Ding's work \cite{D2015, D2016},  many minimal linear codes with $\frac{w_{\rm min}}{w_{\rm max}}>\frac{q-1}{q}$  have been
constructed by selecting the proper defining set or from functions over finite fields (see \cite{DD2015,HY2016,LCXM2018,SLP2016,SGP2017,TLQZ2016,WDX2015,X2016,YY2017,ZLFH2016}). Cohen et al. \cite{CMP2013} provided an example to show that the condition $\frac{w_{\rm min}}{w_{\rm max}}>\frac{q-1}{q}$ in Lemma \ref{Ashikhmin-Barg} is not necessary for a linear code to be minimal.  Recently, Ding, Heng and Zhou \cite{DHZ2018, HDZ2018}  generalized this  sufficient condition and derived a sufficient and necessary condition on all Hamming weights for a given linear code to be
minimal:
\begin{lem}\label{Heng-Ding-Zhou}{$($\rm { Ding-Heng-Zhou}\cite{DHZ2018,HDZ2018})}
A linear code $\mathcal{C}$ over $\mathbb{F}_q$ is minimal if and only if
$$\sum_{c\in\mathbb{F}_q^{*}}\rm{wt}(\mathbf{a}+c\mathbf{b})\neq (q-1)\rm{wt}(\mathbf{a})-\rm{wt}(\mathbf{b})$$
for any $\mathbb{F}_q$-linearly independent codewords $\mathbf{a},\mathbf{b}\in \mathcal{C}$.
\end{lem}

Based on {\bf Lemma \ref{Heng-Ding-Zhou}}, Ding et al. presented three infinite families of minimal binary linear codes with $\frac{w_{\rm min}}{w_{\rm max}}\leq\frac{1}{2}$ in \cite{DHZ2018} and an infinite family of minimal ternary linear codes with $\frac{w_{\rm min}}{w_{\rm max}}<\frac{2}{3}$ in \cite{HDZ2018}, respectively. In \cite{ZYW2018}, Zhang et al. constructed four families of minimal binary linear codes
with $\frac{w_{\rm min}}{w_{\rm max}}\leq\frac{1}{2}$ from Krawtchouk polynomials. Very recently, Bartoli and
Bonini \cite{BB2019} provided infinite families of minimal linear codes; also in \cite{XQ2019} Xu and Qu constructed three classes  of minimal linear codes with $\frac{w_{\rm min}}{w_{\rm max}}<\frac{p-1}{p}$ for any odd prime $p$.

So far, many minimal linear codes are constructed from functions. Let $f:\mathbb{F}_q^m\rightarrow \mathbb{F}_q$ be a function such that $f(\mathbf{x})\neq\mathbf{\omega\cdot x}$ for any $\mathbf{\omega}\in \mathbb{F}_q^m$. Let
$$\mathcal{C}_f=\{\mathbf{c}(u,\mathbf{v})=(uf(\mathbf{x})+\mathbf{v\cdot x})_{\mathbf{x}\in \mathbb{F}_q^m\backslash \{\mathbf{0}\}}\mid\ u\in \mathbb{F}_q,\mathbf{v}\in \mathbb{F}_q^m\}.$$
Then $\mathcal{C}_f$ is a $[q^m-1,m+1]_q$ linear codes. By the choices of the function $f$, many linear codes $\mathcal{C}_f$ can be minimal.

Recently, we gave a new sufficient and necessary condition for a given linear code to be
minimal in \cite{LW2019}. In this paper, first, we  present the new necessary and sufficient conditions for  codewords by a general constructions of linear codes from $q$-ary functions to be minimal; second, based on this sufficient and necessary condition, we give four constructions of minimal linear codes. Our results generalize the constructions in \cite{DHZ2018, HDZ2018, BB2019, MQ2019, MB2019}. In our four constructions, the choices of $f$ are much more flexible than theirs.

The rest of this paper is organized as follows. In Section \ref{section Preliminaries}, we give  basic results on linear codes and our new sufficient and necessary condition. In Section \ref{section minimal codewords}, we  present the new necessary and sufficient conditions for  linear codes constructed from $q$-ary functions to be minimal. In Section \ref{section the first construction}, we give the first construction of minimal linear code which generalizes the constructions in \cite{DHZ2018, HDZ2018, BB2019}. In Section \ref{section the second construction}, we give the second construction of minimal linear code which generalize the construction in \cite{MQ2019}. In Section \ref{M function}, we present a class of minimal linear codes constructed from Maiorana-McFarland functions which generalize the constructions in \cite{DHZ2018}. In Section \ref{polynomials}, we present a class of  minimal linear codes constructed from  polynomials which generalize the construction in \cite{MB2019}. In Section \ref{section Concluding remarks}, we conclude this paper.


\section{\bf Preliminaries}\label{section Preliminaries}
\subsection{Inner product and dual}
Let $k$ be a positive integer. For two vectors $\mathbf{x}=(x_1,x_2,...,x_k)$, $\mathbf{y}=(y_1,y_2,...,y_k)\in  \mathbb{F}_q^k$, their \emph{Euclidean inner product} is:
$$\mathbf{x}\cdot \mathbf{y}:=\mathbf{x}\mathbf{y}^{T}=\sum_{i=1}^k{x}_i{y}_i,$$
where $T$ denotes the transpose operator.

For any subset $S\subseteq \mathbb{F}_q^k$,  we define
$$S^\perp:=\{\mathbf{y}\in \mathbb{F}_q^k\mid \mathbf{y}\cdot\mathbf{x}=0,  \it{{\rm{for\  any}}\  \mathbf{x}\in S}\}.$$
By the definition of $S^\perp$, the following two facts are immediate:\\
1. $S\subseteq (S^\perp)^\perp$;\\
2. If $S$ is a linear subspace of $\mathbb{F}_q^k$, then ${\rm{dim}}(S)+{\rm{dim}}(S^\perp)=k$.

\subsection{Properties of covering}

\begin{lem}\label{cover}
Let $\mathbf{u}=(u_1,...,u_n), \mathbf{v}=(v_1,...,v_n)\in \mathbb{F}_q^n$. Then the following conditions are equivalent:\\
(1) $\mathbf{u}\preceq \mathbf{v}$ $($i.e. ${\rm{Suppt}}(\mathbf{u})\subseteq {\rm{Suppt}}(\mathbf{v}))$;\\
(2) for any $1\leq i\leq n,$ if $u_i\neq 0$, then $v_i\neq 0$;\\
(3) for any $1\leq i\leq n,$ if $v_i= 0$, then $u_i= 0$;\\
(4) \rm{Zero}$(\mathbf{v})\subseteq  \rm{Zero}(\mathbf{u})$, where \rm{Zero}$(\mathbf{u}):=\{1\leq i\leq n: u_i=0\}=[1,...,n]\setminus \rm{Suppt}(\mathbf{u}).$
\end{lem}

\subsection{A sufficient and necessary condition for $q$-ary linear codes to be minimal}
This new sufficient and necessary condition is presented as a main result in \cite{LW2019}, for easier reading, we list the results once more. To present the new sufficient and necessary condition in \cite{LW2019}, some concepts are needed.

Let $k\leq n$ be two positive integers and $q$ a prime power. Let $D:=\{\mathbf{d}_1,...,\mathbf{d}_n\}$  be a multiset and rank$(D)=k$, where $\mathbf{d}_1,...,\mathbf{d}_n\in \mathbb{F}_q^k$. Let  $$\mathcal{C}=\mathcal{C}(D)=\{{\mathbf{c}\mathbf{(x)}}=\mathbf{c}(\mathbf{x};D)=(\mathbf{xd}_1^{T},...,\mathbf{xd}_n^{T}), \mathbf{x}\in \mathbb{F}_q^k\}.$$
Then $\mathcal{C}(D)$ is an $[n,k]_q$ linear code.
For any $\mathbf{y}\in \mathbb{F}_q^k$, we define
$$H(\mathbf{y}):=\mathbf{y}^\perp=\{\mathbf{x}\in \mathbb{F}_q^k\mid\mathbf{xy}^{T}=0\},$$
$$H(\mathbf{y},D):=D\cap H(\mathbf{y})=\{\mathbf{x}\in D\mid\mathbf{xy}^{T}=0\},$$
$$V(\mathbf{y},D):={\rm{Span}}(H(\mathbf{y},D)).$$
It is obvious that $H(\mathbf{y},D)\subseteq V(\mathbf{y},D)\subseteq H(\mathbf{y})$.
The following lemma gives a sufficient and necessary condition for a codeword $\mathbf{c(y)}\in \mathcal{C}(D)$ to be minimal:

\begin{lem}\label{sn}\cite[Theorem 3.2]{LW2019}
Let $\mathbf{y}\in \mathbb{F}_q^k\backslash \{\mathbf{0}\}$. Then the following three conditions are equivalent:\\
$(1)$ $\mathbf{c(y)}$ is minimal in $\mathcal{C}(D)$;\\
$(2)$ \rm{dim}$V(\mathbf{y},D)=k-1$;\\
$(3)$ $V(\mathbf{y},D)=H(\mathbf{y})$.

\end{lem}

The following lemma gives a sufficient and necessary condition for linear codes over $\mathbb{F}_q$ to be minimal.
\begin{lem}\label{sn1}\cite[Theorem 3.3]{LW2019}
The following three conditions are equivalent:\\
$(1)$ $\mathcal{C}(D)$ is minimal;\\
$(2)$ for any $\mathbf{y}\in \mathbb{F}_q^k\backslash \{\mathbf{0}\}$,  \rm{dim}$V(\mathbf{y},D)=k-1$;\\
$(3)$ for any $\mathbf{y}\in \mathbb{F}_q^k\backslash \{\mathbf{0}\}$, $V(\mathbf{y},D)=H(\mathbf{y})$.
\end{lem}

\begin{remark} Almost at the same time, in \cite{ABN2019} and \cite{TQLZ2019}, the authors also obtained the above sufficient and necessary condition independently. They characterize minimal linear codes by cutting blocking sets. Actually, \cite{BB2019}  first used cutting blocking sets to study minimal linear codes.
\end{remark}

\section{\bf  The new sufficient and necessary condition for linear codes constructed from functions to be minimal}\label{section minimal codewords}
\subsection {A general construction of linear code from $q$-ary function}
Let $k\geq 2$ be a positive integer and $m:=k-1$. Let $f:\mathbb{F}_q^m\rightarrow \mathbb{F}_q$ be a function. Define
$$D=D_f:=\{\mathbf{d_x}=(f(\mathbf{x}),\mathbf{x}), \mathbf{x}\in \mathbb{F}_q^m\backslash \{\mathbf{0}\}\},$$ and
$$\mathcal{C}_f:=\mathcal{C}(D)=\mathcal{C}(D_f)=\{\mathbf{c}(u,\mathbf{v}):=((u,\mathbf{v})\cdot\mathbf{d_x})_{\mathbf{x}\in \mathbb{F}_q^m\backslash \{\mathbf{0}\}}=(uf(\mathbf{x})+\mathbf{v\cdot x})_{\mathbf{x}\in \mathbb{F}_q^m\backslash \{\mathbf{0}\}}\mid\ u\in \mathbb{F}_q,\mathbf{v}\in \mathbb{F}_q^m\}.$$

Let $r(D_f)={\rm{rank}}(D_f)$. Then $\mathcal{C}_f$ is a $[q^m-1,r(D_f)]_q$ linear code. It is easy to see that $m\leq r(D_f)\leq m+1$.

\begin{lem}\label{331}
$r(D_f)=m$ if and only if there exists $\mathbf{\omega}\in \mathbb{F}_q^m$, such that $f(\mathbf{x})=\mathbf{\omega\cdot x}.$
\end{lem}

\begin{cor}\label{332}
$r(D_f)=m+1$ if and only if for any given $\mathbf{\omega}\in \mathbb{F}_q^m$,  $f(\mathbf{x})\neq\mathbf{\omega\cdot x}.$
\end{cor}
In the rest of this paper, we assume that $f(\mathbf{x})\neq\mathbf{\omega\cdot x}$ for any
$\mathbf{\omega}\in \mathbb{F}_q^m$, and then $\mathcal{C}_f$ is a $[q^m-1,m+1]_q$ linear code.

\subsection{The sufficient and necessary condition for linear codes constructed from functions to be minimal}

Let $\mathbf{y}=(u,\mathbf{v})\in \mathbb{F}_q^k$. We discuss the following three cases respectively.
\vskip 0.6mm
{\bf Case 1}: When $u\neq 0$ and $\mathbf{v}=\mathbf{0}$, we have the following proposition.
\begin{prop}\label{1}
 When $u\neq 0$ and $\mathbf{v}=\mathbf{0}$, then $\mathbf{c}(u,\mathbf{v})$ is minimal if and only if there exist $\{\mathbf{\alpha}_1,...,\mathbf{\alpha}_m\}$ which is a basis of $\mathbb{F}_q^m$, such that $f(\mathbf{\alpha}_1)=...=f(\mathbf{\alpha}_m)=0.$
\end{prop}
\begin{proof}
For $\mathbf{x}\in \mathbb{F}_q^m\backslash \{\mathbf{0}\}$, $$\mathbf{d_x}\in H(\mathbf{y},D)\Longleftrightarrow uf(\mathbf{x} )+\mathbf{v\cdot x}=0\Longleftrightarrow uf(\mathbf{x}) =0\Longleftrightarrow f(\mathbf{x})=0.$$
By {\bf Lemma} \ref{sn}, $\mathbf{c}(u,\mathbf{v})$ is minimal if and only if there exist $\mathbf{\alpha}_1,...,\mathbf{\alpha}_m\in \mathbb{F}_q^m$, such that $f(\mathbf{\alpha}_1)=...=f(\mathbf{\alpha}_m)=0$ and $\mathbf{d}_{{\alpha}_1},...,\mathbf{d}_{{\alpha}_m}$ are linearly independent over $\mathbb{F}_q$. It is equivalent to that $\{\mathbf{\alpha}_1,...,\mathbf{\alpha}_m \}$ is a basis of $\mathbb{F}_q^m$ and $f({\alpha}_1)=...=f({\alpha}_m)=0$.
\end{proof}
\vskip 0.3mm

{\bf Case 2}: When $u\neq 0$ and $\mathbf{v}\neq \mathbf{0}$, we have the following proposition.
\begin{prop}\label{3}
 When $u\neq 0$ and $\mathbf{v}\neq \mathbf{0}$, let $\mathbf{\omega}=-\frac{1}{u}\mathbf{v}$. Then $\mathbf{c}(u,\mathbf{v})$ is minimal if and only if there exists $\{{\alpha_1},...,{\alpha_m}\}$ which is a basis of $\mathbb{F}_q^m$, such that
$f({\alpha_i})={\omega}\cdot{\alpha_i}$.
\end{prop}
\begin{proof}
For $\mathbf{x}\in \mathbb{F}_q^m\backslash \{\mathbf{0}\}$, $$\mathbf{d_x}\in H(\mathbf{y},D)\Longleftrightarrow uf(\mathbf{x} )+\mathbf{v\cdot x}=0\Longleftrightarrow f(\mathbf{x})=-\frac{1}{u}\mathbf{v\cdot x}\Longleftrightarrow f(\mathbf{x})=\mathbf{\omega}\cdot\mathbf{x}.$$
By {\bf Lemma} \ref{sn}, $\mathbf{c}(u,\mathbf{v})$ is minimal if and only if there exist ${\alpha_1},...,{\alpha_m}\in \mathbb{F}_q^m$, such that for $1\leq i\leq m$, $f({\alpha_i})={\omega}\cdot{\alpha_i}$, and $\mathbf{d}_{\alpha_1},...,\mathbf{d}_{\alpha_m}$ are linearly independent over $\mathbb{F}_q$. Since for $1\leq i\leq m$, $f({\alpha_i})={\omega}\cdot{\alpha_i}$, it is easy to see $\mathbf{d}_{\alpha_1},...,\mathbf{d}_{\alpha_m}$ are linearly independent if and only if ${\alpha_1},...,{\alpha_m}$ are linearly independent over $\mathbb{F}_q$.
\end{proof}
\vskip 0.3mm

{\bf Case 3}: When $u= 0$ and $\mathbf{v}\neq \mathbf{0}$, we have the following proposition.
\begin{prop}\label{4}
 When $u= 0$ and $\mathbf{v}\neq \mathbf{0}$, then $\mathbf{c}(u,\mathbf{v})$ is minimal if and only if there exist ${\alpha_1},...,{\alpha_m}\in H(\mathbf{v})$, such that $\mathbf{d}_{\alpha_1},...,\mathbf{d}_{\alpha_m}$ are linearly independent over $\mathbb{F}_q$.
\end{prop}

\begin{proof}
For $\mathbf{x}\in \mathbb{F}_q^m\backslash \{\mathbf{0}\}$, $$\mathbf{d_x}\in H(\mathbf{y},D)\Longleftrightarrow uf(\mathbf{x} )+\mathbf{v\cdot x}=0\Longleftrightarrow \mathbf{v}\cdot \mathbf{x}\Longleftrightarrow \mathbf{x}\in H(\mathbf{v}).$$
By {\bf Lemma} \ref{sn}, $\mathbf{c}(u,\mathbf{v})$ is minimal if and only if there exist ${\alpha_1},...,{\alpha_m}\in H(\mathbf{v})$, such that  $\mathbf{d}_{\alpha_1},...,\mathbf{d}_{\alpha_m}$ are linearly independent over $\mathbb{F}_q$.
\end{proof}

In the following four sections, we will use the sufficient and necessary condition to construct four classes of minimal linear codes which  generalize the constructions in \cite{DHZ2018, HDZ2018, BB2019, MQ2019, MB2019}. In our four constructions, the choices of $f$ are much more flexible than theirs.

\section{\bf{The first construction of minimal linear codes}}\label{section the first construction}
In this section, we generalize the constructions in \cite{DHZ2018, HDZ2018, BB2019}.
To construct the minimal linear codes, we need the following three lemmas.

\begin{lem}\label{4111}
(1) If $q\geq 3$, then there exists $\{{\alpha_1},...,{\alpha_m}\}$ which is a basis of $\mathbb{F}_q^m$, such that $\mathbf{wt}(\alpha_i)=m$, $1\leq i\leq m;$\\
(2) If $q=2$ and $m$ is even, then there exists $\{{\alpha_1},...,{\alpha_m}\}$ which is a basis of $\mathbb{F}_q^m$ such that $\mathbf{wt}(\alpha_i)\geq m-1$, $1\leq i\leq m;$\\
(3)  If $q=2$ and $m$ is odd, then there exists $\{{\alpha_1},...,{\alpha_m}\}$ which is a basis of $\mathbb{F}_q^m$ such that $\mathbf{wt}(\alpha_i)\geq m-2$, $1\leq i\leq m.$
\end{lem}

\begin{proof}
 Let $\mathbf{1}=(1,...,1)$ and $A=\mathbf{1}^{{T}}\mathbf{1}\in M_{m\times m}(\mathbb{F}_q)$.  \\
 (1) For any $b\in \mathbb{F}_q$ we set
\begin{equation}
bE_m-A={\left[\begin{array}{cccc}
b-1 & -1 & ...&-1\\
-1 & b-1 &...& -1\\
...&...&...&...\\
-1&-1&... & -1
\end{array}
\right ]}={\left[\begin{array}{c}
\alpha_1\\
 \alpha_2\\
...\\
 \alpha_m
\end{array}
\right ]}.
\end{equation}
It is eary to see,
$$|bE_m-A|=b^{m-1}(b-m)\in \mathbb{F}_q.$$

When $q> 3$, we can take $b\neq 0,1,m$, then $\alpha_1,...,\alpha_m$ are linearly independent over $\mathbb{F}_q$ and $\mathbf{wt}(\alpha_i)= m$, $1\leq i\leq m.$

When $q=3$ and $m\equiv 0,1\ ({\rm mod}\ q)$, we can take $b=2$, the result follows.

When  $q=3$ and $m\equiv 2\ ({\rm mod}\ q)$, we consider the following matrix:
\begin{equation}
2E_m-A-2\mathbf{e}_1^{{T}}\mathbf{e}_1={\left[\begin{array}{ccccc}
-1 & -1 & -1&...&-1\\
-1 & 1 &-1&...& -1\\
-1 & -1 &1&...& -1\\
...&...&...&...&...\\
-1&-1&-1&... & 1
\end{array}
\right ]}={\left[\begin{array}{c}
\alpha_1\\
 \alpha_2\\
...\\
 \alpha_m
\end{array}
\right ]}.
\end{equation}

It is easy to see, $$|2E_m-A-2\mathbf{e}_1^{T}\mathbf{e}_1|=-2^{m-1}\neq 0.$$

Hence $\alpha_1,...,\alpha_m$ are linearly independent over $\mathbb{F}_q$ and $\mathbf{wt}(\alpha_i)= m$, $1\leq i\leq m.$

(2) When $q=2$ and $m$ is even, we set
\begin{equation}
E_m-A={\left[\begin{array}{cccc}
0 & -1 & ...&-1\\
-1 & 0 &...& -1\\
...&...&...&...\\
-1&-1&... & 0
\end{array}
\right ]}={\left[\begin{array}{c}
\alpha_1\\
 \alpha_2\\
...\\
 \alpha_m
\end{array}
\right ]}.
\end{equation}
It is eary to see,
$$|E_m-A|=1-m=1\neq 0.$$
Then $\alpha_1,...,\alpha_m$ are linearly independent over $\mathbb{F}_q$ and $\mathbf{wt}(\alpha_i)\geq m-1$, $1\leq i\leq m.$

(3) When $q=2$ and $m$ is odd, we set
\begin{equation}
B={\left[\begin{array}{c}1\\1\\...\\1\end{array}
\right ]}{\left[\begin{array}{cccc}1&...&1&0\end{array}
\right ]}={\left[\begin{array}{cccc}
1 & ... & 1 &0\\
1 & ...& 1&0\\
...&...&...&...\\
1 & ...& 1&0\\
\end{array}
\right ]}.
\end{equation}
Then \begin{equation}
E_m-B={\left[\begin{array}{ccccc}
0 & 1 & ...&1&0\\
1 & 0 &...&1&0\\
...&...&...&...\\
1&1&...&1 &1
\end{array}
\right ]}={\left[\begin{array}{c}
\alpha_1\\
 \alpha_2\\
...\\
 \alpha_m
\end{array}
\right ]}.
\end{equation}
It is eary to see,
$$|E_m-B|=m=1\neq 0.$$

Then $\alpha_1,...,\alpha_m$ are linearly independent over $\mathbb{F}_q$ and $\mathbf{wt}(\alpha_i)\geq m-2$, $1\leq i\leq m.$
\end{proof}

\begin{lem}\label{422}
For any $\omega\in \mathbb{F}_q^m\backslash \{\mathbf{0}\}$, there exists $\{{\beta_1},...,{\beta_m}\}$ which is a basis of $\mathbb{F}_q^m$, such that $1\leq\mathbf{wt}(\beta_i)\leq 2$ and $\omega\cdot \beta_i=1$, $1\leq i\leq m.$
\end{lem}
\begin{proof}
Let $\omega=(w_1,...,w_m)$. Since $\omega\neq \mathbf{0}$, there exists $w_{i_0}\neq 0$. Let
$$
 \beta_i:=\left\{
            \begin{array}{ll}
             w_{i_0}^{-1}{\bf e}_{i_0}, & \hbox{if\ $i=i_0;$}\\
              {\bf e}_{i}+ w_{i_0}^{-1}(1-w_i){\bf e}_{i_0}, & \hbox{if\ $i\neq i_0.$}
           \end{array}
          \right.
$$
Then $1\leq\mathbf{wt}(\beta_i)\leq 2$ and $\omega\cdot \beta_i=1$.

It is easy to see that  ${\beta_1},...,{\beta_m}$ are linearly independent and they constitute a basis of  $\mathbb{F}_q^m$.
\end{proof}

\begin{lem}\label{43}
For any $\mathbf{v}\in \mathbb{F}_q^m\backslash \{\mathbf{0}\}$, there exists $\{{\alpha_1},...,{\alpha_{m-1}}\}$ which is a basis of $H(\mathbf{v})$, such that $1\leq\mathbf{wt}(\alpha_i)\leq 2$, $1\leq i\leq m-1.$
\end{lem}
\begin{proof}
Let $\mathbf{v}=(v_1,...,v_m)$. Since $\mathbf{v}\neq \mathbf{0}$, there exists $v_{i_0}\neq 0$. For $i\neq i_0$, let
 $$\beta_i={\bf e}_{i}-v_{i_0}^{-1}v_i{\bf e}_{i_0}.$$
 Then $\beta_i\in H(\mathbf{v})$ and $1\leq\mathbf{wt}(\beta_i)\leq 2$. It is easy to see that $\{{\beta_i}:i\neq i_0\}$ are linearly independent. For $1\leq i\leq m-1$, let
$$
 \alpha_i:=\left\{
            \begin{array}{ll}
             \beta_i, & \hbox{if\ $i< i_0;$}\\
              \beta_{i+1}, & \hbox{if\ $i\geq i_0.$}
           \end{array}
          \right.
$$
Then $\{{\alpha_1},...,{\alpha_{m-1}}\}$  is a basis of $H(\mathbf{v})$  and $1\leq\mathbf{wt}(\alpha_i)\leq 2$, $1\leq i\leq m-1.$
\end{proof}

Now we start to construct the minimal linear codes.
\begin{thm}\label{40}
Let $q>2$, $m\geq 3$ and $f:\mathbb{F}_q^m\rightarrow \mathbb{F}_q$ be a function. If $f(\mathbf{x})$ satisfies the following two conditions:\\
(1) if  $1\leq \mathbf{wt(x)}\leq 2$, then for any $a\in \mathbb{F}_q^*$, $f(a\mathbf{x})=f(\mathbf{x})\neq 0$;\\
(2) if  $\mathbf{wt(x)}=m,$ then $f(\mathbf{x})=0$,\\
then  $\mathcal{C}_f=\mathcal{C}(D_f)$ is minimal.
\end{thm}
\begin{proof}
By condition (1) in this theorem, it is easy to see $f(\mathbf{x})\neq\mathbf{\omega\cdot x}$ for any $\mathbf{\omega}\in \mathbb{F}_q^m$, then by {\bf Corollary \ref{332}}, $\mathcal{C}_f$ is a $[q^m-1,m+1]_q$ linear code.

{\bf Case 1:} When $u\neq 0$ and $\mathbf{v}=\mathbf{0}$, by {\bf Lemma \ref{4111}}(1), there exists  $\{{\alpha_1},...,{\alpha_m}\}$ which is a basis of $\mathbb{F}_q^m$ such that $\mathbf{wt}(\alpha_i)=m$, $1\leq i\leq m.$ By condition (2),  we can get $f(\alpha_i)=0$,  $1\leq i\leq m.$ By {\bf Proposition \ref{1}}, $\mathbf{c}(u,\mathbf{v})$ is minimal.

{\bf Case 2:} When $u\neq 0$ and $\mathbf{v}\neq \mathbf{0}$, let $\mathbf{\omega}=-\frac{1}{u}\mathbf{v}$. By {\bf Lemma \ref{422}}, there exist $\{{\beta_1},...,{\beta_m}\}$ which is a basis of $\mathbb{F}_q^m$, such that $1\leq\mathbf{wt}(\beta_i)\leq 2$ and $\omega\cdot \beta_i=1$, $1\leq i\leq m.$ Let $\alpha_i=f(\beta_i)\beta_i$. By condition (1), we have $f(\beta_i)\neq 0$, then $1\leq\mathbf{wt}(\alpha_i)\leq 2$ and $\{{\alpha_1},...,{\alpha_m}\}$ is a basis of $\mathbb{F}_q^m$. Thus by  condition (1) in  this theorem, we have $f(\alpha_i)=f(\beta_i)=f(\beta_i)(\omega\cdot \beta_i)=\omega \cdot (f(\beta_i)\beta_i)=\omega \cdot \alpha_i$. By {\bf Proposition \ref{3}}, $\mathbf{c}(u,\mathbf{v})$ is minimal.

{\bf Case 3:}  When $u= 0$ and $\mathbf{v}\neq \mathbf{0}$,  by {\bf Lemma \ref{43}}, there exists  $\{{\alpha_1},...,{\alpha_{m-1}}\}$ which is a basis of $H(\mathbf{v})$, such that $1\leq\mathbf{wt}(\alpha_i)\leq 2$, $1\leq i\leq m-1.$ Let $\alpha_m =a\alpha_1$, $a\in \mathbb{F}_q\setminus\{0,1\}$. Then $\alpha_m\in H(\mathbf{v})$.
Now we prove that  $\mathbf{d}_{\alpha_1},...,\mathbf{d}_{\alpha_m}$ are linearly independent. Let $\sum_{i=1}^m k_i\mathbf{d}_{\alpha_i}=\mathbf{0}$. It is equivalent to
$(\sum_{i=1}^m k_i f({\alpha_i}), \sum_{i=1}^m k_i \alpha_i)=(0, \mathbf{0})$. We get
$$\mathbf{0}=\sum_{i=1}^{m} k_i \alpha_i=\sum_{i=2}^{m-1} k_i \alpha_i+(k_1+ak_m)\alpha_1.$$
Since $\alpha_1,...,\alpha_{m-1}$ are linearly independent, $k_i=0$, $i=2,...,m-1$ and $k_1+ak_m=0$.
Then $$0=\sum_{i=1}^m k_i f({\alpha_i})=k_1f(\alpha_1)+k_mf(\alpha_m)=k_1f(\alpha_1)+k_mf(a\alpha_1).$$
Since $1\leq\mathbf{wt}(\alpha_1)\leq 2$ and $a\neq 0,$ by condition (1), we get $f(a\alpha_1)=f(\alpha_1)$, and then
$$0=k_1f(\alpha_1)+k_mf(\alpha_1)=k_m(1-a)f(\alpha_1)$$
Since $a\neq 1$ and $f(\alpha_1)\neq 0$, we get $k_m=0$ and $k_1=0$. Thus $\mathbf{d}_{\alpha_1},...,\mathbf{d}_{\alpha_m}$ are linearly independent.
By {\bf Proposition \ref{4}}, $\mathbf{c}(u,\mathbf{v})$ is minimal.

This completes the proof.
\end{proof}

\begin{thm}\label{44}
Let $q=2$, $m\geq 4$ and $f:\mathbb{F}_q^m\rightarrow \mathbb{F}_q$ be a function. If $f(\mathbf{x})$ satisfies the following two conditions:\\
(1) when  $1\leq \mathbf{wt(x)}\leq 2$,  $f(\mathbf{x})\neq 0, (i.e. f(\mathbf{x})=1)$;\\
(2) when $m$ is even and $\mathbf{wt(x)}\geq m-1$ or when $m$ is odd and $\mathbf{wt(x)}\geq m-2,\ f(\mathbf{x})=0$,\\
then  $\mathcal{C}_f=\mathcal{C}(D_f)$ is minimal.
\end{thm}
\begin{proof}

{\bf Case 1:} When $u\neq 0$ and $\mathbf{v}=\mathbf{0}$. If $m$ is even,  by {\bf Lemma \ref{4111}}(2), there exists  $\{{\alpha_1},...,{\alpha_{m}}\}$ which is a basis of $\mathbb{F}_q^m$ satisfy $\mathbf{wt}(\alpha_i)\geq m-1$, $1\leq i\leq m.$ By condition (2),  we can get $f(\alpha_i)=0$,  $1\leq i\leq m.$
 If $m$ is odd,  by {\bf Lemma \ref{4111}}(3), there exists  $\{{\alpha_1},...,{\alpha_{m}}\}$ which is a basis of $\mathbb{F}_q^m$ satisfy $\mathbf{wt}(\alpha_i)\geq m-2$, $1\leq i\leq m.$ By condition (2),  we can get $f(\alpha_i)=0$,  $1\leq i\leq m.$
 Then by {\bf Proposition \ref{1}}, $\mathbf{c}(u,\mathbf{v})$ is minimal.

{\bf Case 2:} When $u\neq 0$ and $\mathbf{v}\neq \mathbf{0}$, let $\mathbf{\omega}=-\frac{1}{u}\mathbf{v}$. By {\bf Lemma \ref{422}}, there exists $\{{\alpha_1},...,{\alpha_m}\}$ which is a basis of $\mathbb{F}_q^m$ such that $1\leq\mathbf{wt}(\alpha_i)\leq 2$ and $\omega\cdot \alpha_i=1$, $1\leq i\leq m.$ By condition (1), we have $f(\alpha_i)=1=\omega\cdot \alpha_i$, $1\leq i\leq m.$ By {\bf Proposition \ref{3}}, $\mathbf{c}(u,\mathbf{v})$ is minimal.

{\bf Case 3:}  When $u= 0$ and $\mathbf{v}\neq \mathbf{0}$, Let $\mathbf{v}=(v_1,...,v_m)$. Then there exists $v_{i_0}\neq 0$. For $i\neq i_0$, let
 $$\alpha_i={\bf e}_{i}-v_{i_0}^{-1}v_i{\bf e}_{i_0}.$$
 Then $\{\alpha_i: \ i\neq i_0\}$ is a basis of $H(\mathbf{v})$ and $1\leq \mathbf{wt}(\alpha_i)\leq 2$.

 When $m$ is even, let $\alpha_{i_0}=\Sigma_{i\neq i_0}\alpha_i.$ Then $\alpha_{i_0}\in H(\mathbf{v})$ and $\mathbf{wt}(\alpha_{i_0})\geq m-1.$ By condition (2), $f(\alpha_{i_0})=0$.  Now we prove that  $\mathbf{d}_{\alpha_1},...,\mathbf{d}_{\alpha_m}$ are linearly independent. Let $\sum_{i=1}^m k_i\mathbf{d}_{\alpha_i}=\mathbf{0}$. It is equivalent to
$(\sum_{i=1}^m k_i f({\alpha_i}), \sum_{i=1}^m k_i \alpha_i)=(0, \mathbf{0})$. We get
$$\mathbf{0}=\sum_{i=1}^m k_i \alpha_i=\sum_{i\neq{i_0}}(k_i+k_{i_0}) \alpha_i.$$
Since $\{\alpha_i: \ i\neq i_0\}$ are linearly independent, $k_i=-k_{i_0}=k_{i_0}$, for $i\neq i_0$. Since
$$
 f(\alpha_i)=\left\{
            \begin{array}{ll}
             0, & \hbox{if\ $i= i_0;$}\\
             1, & \hbox{if\ $i\neq i_0,$}
           \end{array}
          \right.
$$
$$0=\sum_{i=1}^m k_i f({\alpha_i})=\sum_{i\neq i_0} k_if({\alpha_i})+k_{i_0}f({\alpha_{i_0}})=(m-1)k_{i_0}=k_{i_0}.$$
Thus for all $i$, $k_i=k_{i_0}=0$,  $\mathbf{d}_{\alpha_1},...,\mathbf{d}_{\alpha_m}$ are linearly independent.
By {\bf Proposition \ref{4}}, $\mathbf{c}(u,\mathbf{v})$ is minimal.

When $m$ is odd, let $i_1\neq i_0$ and $\alpha_{i_0}=\Sigma_{i\neq i_0,i_1}\alpha_i.$ Then $\alpha_{i_0}\in H(\mathbf{v})$ and $\mathbf{wt}(\alpha_{i_0})\geq m-2.$ By condition (2), $f(\alpha_{i_0})=0$.
Now we prove that  $\mathbf{d}_{\alpha_1},...,\mathbf{d}_{\alpha_m}$ are linearly independent. Let $\sum_{i=1}^m k_i\mathbf{d}_{\alpha_i}=\mathbf{0}$, it is equivalent to
$(\sum_{i=1}^m k_i f({\alpha_i}), \sum_{i=1}^m k_i \alpha_i)=(0, \mathbf{0})$. We get
$$\mathbf{0}=\sum_{i=1}^m k_i \alpha_i=\sum_{i\neq{i_0},{i_1}}(k_i+k_{i_0}) \alpha_i+k_{i_1} \alpha_{i_1}.$$
Since $\{\alpha_i: \ i\neq i_0\}$ are linearly independent, for $i\neq i_1$, $k_i=-k_{i_0}=k_{i_0}$, and $k_{i_1}=0$.  Since
$$
 f(\alpha_i)=\left\{
            \begin{array}{ll}
             0, & \hbox{if\ $i= i_0;$}\\
             1, & \hbox{if\ $i\neq i_0,$}
           \end{array}
          \right.
$$
$$0=\sum_{i=1}^m k_i f({\alpha_i})=\sum_{i\neq i_0,i_1} k_if({\alpha_i})+k_{i_0}f({\alpha_{i_0}})+k_{i_1}f({\alpha_{i_1}})=(m-2)k_{i_0}=k_{i_0}.$$
Thus  $k_i=0$, $1\leq i\leq m$,  $\mathbf{d}_{\alpha_1},...,\mathbf{d}_{\alpha_m}$ are linearly independent.
By {\bf Proposition \ref{4}}, $\mathbf{c}(u,\mathbf{v})$ is minimal.

This completes the proof.
\end{proof}

As the special cases of {\bf Theorem \ref{40}} and {\bf \ref{44}}, we have the following corollaries.
\begin{cor}\label{45}
Let $m$, $t$ be integers with $m\geq 7$ and $2\leq t\leq \lfloor\frac{m-3}{2}\rfloor$. Assume that  $f:\mathbb{F}_2^m\rightarrow \mathbb{F}_2$ is a function defined as follows:
$$
 f(\mathbf{x})=\left\{
            \begin{array}{ll}
            1, & \hbox{$1\leq {\bf wt}({\bf x})\leq t$},\\
            0, & \hbox{${\bf wt}({\bf x})>t.$}
           \end{array}
          \right.
$$
Then $\mathcal{C}_f$ is minimal.
\end{cor}
\begin{remark} The $\mathcal{C}_f$ in {\bf Corollary \ref{45}} is the one in \cite[Theorem 3.1]{DHZ2018}.
\end{remark}

\begin{cor}\label{46}
Let $m$, $t$ be integers with $m\geq 5$ and $2\leq t\leq \lfloor\frac{m-1}{2}\rfloor$. Assume that  $f:\mathbb{F}_3^m\rightarrow \mathbb{F}_3$ is a function defined as follows:
$$
 f(\mathbf{x})=\left\{
            \begin{array}{ll}
            1, & \hbox{$1\leq {\bf wt}({\bf x})\leq t$},\\
            0, & \hbox{${\bf wt}({\bf x})>t.$}
           \end{array}
          \right.
$$
Then $\mathcal{C}_f$ is minimal.
\end{cor}
\begin{remark} The $\mathcal{C}_f$ in {\bf Corollary \ref{46}} is the one in \cite[Theorem 18]{HDZ2018}.\end{remark}

\begin{cor}\label{47}
Let $m$, $t$ be integers with $m> 3$ and $2\leq t\leq m-2$.   Assume that  $f:\mathbb{F}_q^m\rightarrow \mathbb{F}_q$ is a function defined as follows:
$$
 f(\mathbf{x})=\left\{
            \begin{array}{ll}
            a_i\neq 0, &1\leq \hbox{${\bf wt}({\bf x})=i\leq t$},\\
            0, & \hbox{${\bf wt}({\bf x})>t.$}
           \end{array}
          \right.
$$
Then $\mathcal{C}_f$ is minimal.
\end{cor}
\begin{remark} The $\mathcal{C}_f$ in {\bf Corollary \ref{47}} is the one in \cite[ Theorem III 2]{BB2019}. \end{remark}

It is easy  to see that our constructions contain the above three constructions in \cite{DHZ2018, HDZ2018, BB2019} as special cases,  so our constructions are more general.
\begin{exm}
Let $q=3$, $m=4$, $\mathbf{x}=(x_1,x_2,x_3,x_4)$. Then $n=q^m-1=80$ and $k=m+1=5$. With the help of Magma, we get the following results.

(1) Let
$$
 f_1(\mathbf{x})=\left\{
            \begin{array}{ll}
            1, & \hbox{$1\leq {\bf wt}({\bf x})\leq 2$},\\
            0, & \hbox{${\bf wt}({\bf x})>2.$}
           \end{array}
          \right.
$$

 Then  $\mathcal{C}_{f_1}$ is a minimal linear code over $\mathbb{F}_3$ with parameters $[80,\ 5,\ 32]$ and weight enumerator
 $$1+2z^{32}+64z^{50}+48z^{53}+80z^{54}+32z^{56}+16z^{65}.$$
  Furthermore, $\frac{w_{\rm min}}{w_{\rm max}}=\frac{32}{65}<\frac{2}{3}.$

 (2) Let
$$
 f_2(\mathbf{x})=\left\{
            \begin{array}{ll}
            1, & \hbox{$1\leq {\bf wt}({\bf x})\leq 2$},\\
            x_1, & \hbox{${\bf wt}({\bf x})=3,$}\\
            0, & \hbox{${\bf wt}({\bf x})=4.$}
           \end{array}
          \right.
$$

 Then  $\mathcal{C}_{f_2}$ is a minimal linear code over $\mathbb{F}_3$ with parameters $[80,\ 5,\ 41]$ and weight enumerator
 $$1+2z^{41}+24z^{47}+40z^{50}+24z^{53}+80z^{54}+58z^{56}+14z^{65}.$$
 Furthermore, $\frac{w_{\rm min}}{w_{\rm max}}=\frac{41}{65}<\frac{2}{3}.$

 Since $m=4<5$, $\mathcal{C}_{f_1}$ and $\mathcal{C}_{f_2}$ are not contained in the construction in {\bf Corollary \ref{46}}. When $m=4$, $\mathcal{C}_{f_1}$ is the unique linear code determined by {\bf Corollary \ref{47}} in the sense of equivalence. So $\mathcal{C}_{f_2}$ is a new linear code constructed by {\bf Theorem \ref{40}}, the parameters of $\mathcal{C}_{f_1}$ and $\mathcal{C}_{f_2}$ are different, and the minimum distance of $\mathcal{C}_{f_2}$ is better than that of $\mathcal{C}_{f_1}$. In fact, the choices of $f$ in {\bf Theorem \ref{40}} and {\bf Theorem \ref{44}} are very flexible.

\end{exm}

\section{\bf{The second onstruction of minimal linear codes}}\label{section the second construction}
In this section we generalize a construction in \cite{MQ2019}.
\begin{thm}\label{51}
Let $f:\mathbb{F}_q^m\rightarrow \mathbb{F}_q$ be a function. If $f(\mathbf{x})$ satisfies the following two conditions:\\
(1) when  $1\leq \mathbf{wt(x)}\leq 2$, $f(\mathbf{x})=0$;\\
(2) when  $\mathbf{wt(x)}\geq m-1$, then for any $a\in \mathbb{F}_q^*$, $f(a\mathbf{x})=f(\mathbf{x})\neq 0$,\\
then  $\mathcal{C}_f=\mathcal{C}(D_f)$ is minimal.
\end{thm}
\begin{proof}
It is easy to see that $f(\mathbf{x})\neq\mathbf{\omega\cdot x}$ for any $\mathbf{\omega}\in \mathbb{F}_q^m$, then by {\bf Corollary \ref{332}}, $\mathcal{C}_f$ is a $[q^m-1,m+1]_q$ linear code.

{\bf Case 1:} By condition(1), we get $f(\mathbf{e}_i)=0$, $1\leq i\leq m.$ Then $\{\mathbf{e}_1,...,\mathbf{e}_m\}$ is a basis of  $\mathbb{F}_q^m$ and $f(\mathbf{e}_1)=...=f(\mathbf{e}_m)=0$. By {\bf Proposition \ref{1}},  $\mathbf{c}(u,\mathbf{v})$ is minimal.

{\bf Case 2:} When $u\neq 0$ and $\mathbf{v}\neq \mathbf{0}$, let $\mathbf{\omega}=(w_1,...,w_m)=-\frac{1}{u}\mathbf{v}$. Then there exists $w_{i_0}\neq 0.$ For $i\neq i_0$, let $\alpha_i=\mathbf{e}_i-w_iw_{i_0}^{-1}\mathbf{e}_{i_0}.$ It is easy to see $\{\alpha_i,\ i\neq i_0\}$ is a basis of $H(\omega)$, $\omega\cdot \alpha_i=0$ and $1\leq \mathbf{wt}(\alpha_i)\leq 2$. By condition (1), $f(\alpha_i)=0=\omega\cdot \alpha_i$.

Let $\beta_{i_0}=\Sigma_{i\neq i_0}\alpha_i+w_{i_0}^{-1}\mathbf{e}_{i_0}$. Then $\omega\cdot \beta_{i_0}=1$ and $\mathbf{wt}(\beta_{i_0})\geq m-1.$ By condition (2) of this theorem, we get $f(\beta_{i_0})\neq 0.$ Let $\alpha_{i_0}=f(\beta_{i_0})\beta_{i_0}$. By condition (2), we get $f(\alpha_{i_0})=f(f(\beta_{i_0})\beta_{i_0})=f(\beta_{i_0})=f(\beta_{i_0})(\omega\cdot \beta_{i_0})=\omega\cdot(f(\beta_{i_0})\beta_{i_0})=\omega\cdot \alpha_{i_0}.$ It is easy to see that $\alpha_1,...,\alpha_m$ are linearly independent. Thus we find that $\{\alpha_1,...,\alpha_m\}$ is a basis of $\mathbb{F}_q^m$, and $f(\alpha_{i})=\omega\cdot \alpha_{i}$. By {\bf Proposition \ref{3}},  $\mathbf{c}(u,\mathbf{v})$ is minimal.

{\bf Case 3:}  When $u= 0$ and $\mathbf{v}\neq \mathbf{0}$, let $\mathbf{v}=(v_1,...,v_m)$. Then there exists $v_{i_0}\neq 0.$ For $i\neq i_0$, let $\alpha_i=\mathbf{e}_i-v_iv_{i_0}^{-1}\mathbf{e}_{i_0}.$ Then $\{\alpha_i:i\neq i_0\}$ is a basis of $H(\mathbf{v})$ and $1\leq \mathbf{wt}(\alpha_i)\leq 2$. By condition (1) of this theorem, we get $f(\alpha_i)=0$. Let $\alpha_{i_0}=\Sigma_{i\neq i_0}\alpha_i.$
Then $\alpha_{i_0}\in H(\mathbf{v})$ and $\mathbf{wt}(\alpha_{i_0})\geq m-1.$ By condition(2), we get $f(\alpha_i)\neq 0$. Thus it is easy to see that $\alpha_i\in H(\mathbf{v}),\ 1\leq i\leq m$ and $\mathbf{d}_{\alpha_1},...,\mathbf{d}_{\alpha_m}$ are linearly independent. By {\bf Proposition \ref{4}},  $\mathbf{c}(u,\mathbf{v})$ is minimal.

This completes the proof.
\end{proof}

As a special case of {\bf Theorem \ref{51}}, we have the following corollary.
\begin{cor}\label{52}
Let $m$, $t$ be integers with $m> 3$ and $2\leq t\leq m-2$. Assume that  $f:\mathbb{F}_2^m\rightarrow \mathbb{F}_2$ is a function defined as follows:
$$
 f(\mathbf{x})=\left\{
            \begin{array}{ll}
            0, & \hbox{${\bf wt}({\bf x})\leq t$},\\
            1, & \hbox{${\bf wt}({\bf x})>t.$}
           \end{array}
          \right.
$$
Then $\mathcal{C}_f$ is minimal.
\end{cor}
\begin{remark} The $\mathcal{C}_f$ in {\bf Corollary \ref{47}} is the one in \cite[Theorem 3.6]{MQ2019}. It is easy to see that our construction contains \cite[Theorem 3.6]{MQ2019}  as a special case, and our construction is more general.

\end{remark}

\begin{exm}
Let $q=2$, $m=5$, $\mathbf{x}=(x_1,x_2,x_3,x_4)$. Then $n=q^m-1=31$ and $k=m+1=6$. With the help of Magma, we get the following results.

(1) Let
$$
 f_1(\mathbf{x})=\left\{
            \begin{array}{ll}
            0, & \hbox{$1\leq {\bf wt}({\bf x})\leq 2$},\\
            1, & \hbox{${\bf wt}({\bf x})>2.$}
           \end{array}
          \right.
$$

 Then  $\mathcal{C}_{f_1}$ is a minimal linear code over $\mathbb{F}_2$ with parameters $[31,\ 6,\ 10]$ and weight enumerator
 $$1+6z^{10}+47z^{16}+10z^{18}.$$
  Furthermore, $\frac{w_{\rm min}}{w_{\rm max}}=\frac{10}{18}>\frac{1}{2}.$

 (2) Let
$$
 f_2(\mathbf{x})=\left\{
            \begin{array}{ll}
            0, & \hbox{$1\leq {\bf wt}({\bf x})\leq 3$},\\
            1, & \hbox{${\bf wt}({\bf x})>3.$}
           \end{array}
          \right.
$$

 Then  $\mathcal{C}_{f_1}$ is a minimal linear code over $\mathbb{F}_2$ with parameters $[31,\ 6,\ 6]$ and weight enumerator
 $$1+z^{6}+5z^{12}+5z^{14}+41z^{16}+10z^{18}+z^{20}.$$
  Furthermore, $\frac{w_{\rm min}}{w_{\rm max}}=\frac{6}{20}<\frac{1}{2}.$

 (3) Let
$$
 f_3(\mathbf{x})=\left\{
            \begin{array}{ll}
            0, & \hbox{$1\leq {\bf wt}({\bf x})\leq 2$},\\
            x_1+x_2, & \hbox{${\bf wt}({\bf x})=3,$}\\
            1, & \hbox{${\bf wt}({\bf x})\geq4.$}
           \end{array}
          \right.
$$
\end{exm}
 Then  $\mathcal{C}_{f_3}$ is a minimal linear code over $\mathbb{F}_2$ with parameters $[31,\ 6,\ 10]$ and weight enumerator
 $$1+3z^{10}+4z^{12}+3z^{14}+43z^{16}+9z^{18}+z^{22}.$$
 Furthermore, $\frac{w_{\rm min}}{w_{\rm max}}=\frac{10}{22}<\frac{1}{2}.$

  When $m=5$, $f_1$ and $f_2$ are all the functions determined by {\bf Corollary \ref{52}} and $f_3$ is a new function constructed by {\bf Theorem \ref{51}}. The parameters of $\mathcal{C}_{f_1}$, $\mathcal{C}_{f_2}$ and $\mathcal{C}_{f_3}$ are different, so $\mathcal{C}_{f_3}$ is a new minimal linear code. In fact, the choices of $f$ in {\bf Theorem \ref{51}} are very flexible.

In the following two sections, let
 $$\mathbf{e}_1=(1,0,0,...,0,0),$$
 $$\mathbf{e}_2=(0,1,0,...,0,0),$$
 $$......$$
 $$\mathbf{e}_l=(0,0,0,...,0,1),$$
denote the  standard basis of the $l$-dimensional vector space $\mathbb{F}_q^l$.

\section{\bf Minimal linear codes constructed from  Maiorana-McFarland functions}\label{M function}
In this section, we present a construction of minimal linear codes from Maiorana-McFarland functions. Our construction generalizes the result in \cite{DHZ2018}.
Let $s$ and $t$ be positive integers, $m=s+t$. Then $\mathbb{F}_q^m=\mathbb{F}_q^s\bigoplus \mathbb{F}_q^t$. Let $\alpha=(\beta,\gamma)\in \mathbb{F}_q^m$, where $\beta\in \mathbb{F}_q^s$ and $\gamma\in \mathbb{F}_q^t$. Let $\phi$ be a function from $\mathbb{F}_q^s\rightarrow \mathbb{F}_q^t$ and $g$ a function from $\mathbb{F}_q^s\rightarrow \mathbb{F}_q$.

A Maiorana-McFarland function $f$ is defined from $\mathbb{F}_q^m\rightarrow \mathbb{F}_q$ satisfying
 $$f(\alpha)=f(\beta, \gamma)=\phi(\beta)\cdot \gamma+g(\beta),$$ and
$$\mathcal{C}_f:=\{\mathbf{c}(u,\mathbf{v}):=((u,\mathbf{v})\cdot\mathbf{d_{\alpha}})_{\alpha\in \mathbb{F}_q^m\backslash \{\mathbf{0}\}}=(uf(\mathbf{\alpha})+\mathbf{v\cdot \alpha})_{\alpha\in \mathbb{F}_q^m\backslash \{\mathbf{0}\}}\mid\ u\in \mathbb{F}_q,\mathbf{v}\in \mathbb{F}_q^m\}.$$

\begin{lem}\label{311}
Let $A\mathbf{x}=\mathbf{b},$ where $A\in \mathbb{F}_q^{l\times n}$ and $\mathbf{b}\in \mathbb{F}_q^{l\times 1}$. Then\\
$(1)$ when $\mathbf{b}=\mathbf{0},$  there exist $n-rank(A)$ linearly independent solutions of the equations, \\
$(2)$ when $\mathbf{b}\neq\mathbf{0}$ and $rank(A)=rank(A,\mathbf{b})$, there exist $n-rank(A)+1$ linearly independent solutions of the equations.
\end{lem}

According to $q>2$ and $q=2$, we have the following two theorems.
\begin{thm}\label{41}
Let $q> 2$, $s\geq 2$ and $t\geq 2$. Let $U=\{\beta\in \mathbb{F}_q^s: \mathbf{wt}(\beta)\leq 1\}.$ If $\phi$ and $g$ satisfy:\\
$(1)$ $\phi$ is an injection from $U$ to $\mathbb{F}_q^t\backslash \{\mathbf{0}\}$,\\
$(2)$ $g(\beta)\equiv c\in \mathbb{F}_q^*,$ for all $\beta\in \mathbb{F}_q^s,$\\
then $\mathcal{C}_f$ is a $[q^m-1, m+1]_q$ minimal linear code.

\end{thm}
\begin{proof} By condition (2), it is easy to see that $f(\mathbf{x})\neq\mathbf{\omega\cdot x}$ for any $\mathbf{\omega}\in \mathbb{F}_q^m$, then by {\bf Corollary \ref{332}}, $\mathcal{C}_f$ is a $[q^m-1,m+1]_q$ linear code.
Next, we will prove the codewords $\mathbf{c}(u,\mathbf{v})$ are minimal in three different cases.

{\bf Case 1:} When $u\neq 0$ and $\mathbf{v}=\mathbf{0}$, since $\phi(\mathbf{0})\neq \mathbf{0}$ and $g(\mathbf{0})\neq0$, by {\bf Lemma \ref{311}} there exists ${\gamma_1},...,{\gamma_t}\in \mathbb{F}_q^t$ which  are linearly independent, such that $\phi(\mathbf{0})\cdot \gamma_j=-g(\mathbf{0})$, $1\leq j\leq t$. Since for $1\leq i\leq s,\ \phi(\mathbf{e}_i)\neq \mathbf{0},$ by {\bf Lemma \ref{311}}, there exist $\gamma_{i0}\in \mathbb{F}_q^t$ such that $\phi(\mathbf{e}_i)\cdot \gamma_{i0}=-g(\mathbf{e}_i)$, where $\mathbf{e}_i\in \mathbb{F}_q^s$, $1\leq i\leq s$. Let
   $$
   \alpha_k=\left\{
            \begin{array}{ll}
             (\mathbf{0},\gamma_k), & \hbox{$1\leq k\leq t$},\\
             (\mathbf{e}_{k-t},\gamma_{(k-t)0}), & \hbox{$t<k\leq m$}.
           \end{array}
          \right.
   $$
Then for $1\leq k\leq m$, $f(\alpha_k)=0$ and $\{\alpha_1,...,\alpha_m\}$
   is a basis of $\mathbb{F}_q^m$. By {\bf Proposition \ref{1}}, $\mathbf{c}(u,\mathbf{v})$ is minimal.

{\bf Case 2:} When $u=0$ and $\mathbf{v}\neq \mathbf{0}$, let $\mathbf{v}=(\mathbf{v_1},\mathbf{v_2})$, $\mathbf{v_1}\in \mathbb{F}_q^s$, $\mathbf{v_2}\in \mathbb{F}_q^t$.

 {\bf If $\mathbf{v_2}\neq 0$}, by {\bf Lemma \ref{311}}, there exist $\gamma_1,...,\gamma_{t-1}\in \mathbb{F}_q^t$ which are linear independent, such that $\mathbf{v_2}\cdot \gamma_j=0=-\mathbf{v_1}\cdot \mathbf{0},\ 1\leq j\leq t-1,$ and  there exist $\gamma_{i0}\in \mathbb{F}_q^t$ such that $\mathbf{v_2}\cdot \gamma_{i0}=-\mathbf{v_1}\cdot\mathbf{e}_i$, where $\mathbf{e}_i\in \mathbb{F}_q^s$, $1\leq i\leq s$. Let
  $$
   \alpha_k=\left\{
            \begin{array}{ll}
             (\mathbf{0},\gamma_k), & \hbox{$1\leq k\leq t-1$},\\
             (\mathbf{e}_{k-(t-1)},\gamma_{(k-(t-1))0}), & \hbox{$t-1< k\leq m-1$},\\
             (\mathbf{0},\mathbf{0}), & \hbox{$k=m$}.
           \end{array}
          \right.
   $$
   Then $\alpha_1,...,\alpha_m\in H(\mathbf{v})$. Since $f(\mathbf{0},\mathbf{0})=g(\mathbf{0})\neq 0,$ it is easy to prove that $\mathbf{d}_{\alpha_1},...,\mathbf{d}_{\alpha_m}$ are linearly independent. By {\bf Proposition \ref{4}}, $\mathbf{c}(u,\mathbf{v})$ is minimal.

   {\bf If $\mathbf{v_2}=\mathbf{0}$}, then $\mathbf{v_1}\neq \mathbf{0}$. By {\bf Lemma \ref{311}} there exist $\beta_1,...,\beta_{s-1}\in \mathbb{F}_q^s$ which are linearly independent, such that $\mathbf{v_1}\cdot \beta_i=0=-\mathbf{v_2}\cdot \mathbf{0}.$ Let

   $$
   \alpha_k=\left\{
            \begin{array}{ll}
             (\beta_k,\mathbf{0}), & \hbox{$1\leq k\leq s-1$},\\
             (\mathbf{0},\mathbf{e}_{k-(s-1)}), & \hbox{$s-1< k\leq m-1$},\\
             (\mathbf{0},\mathbf{0}), & \hbox{$k=m$},
           \end{array}
          \right.
   $$
   where $\mathbf{e}_{k-(s-1)}\in \mathbb{F}_q^t.$ Then $\alpha_1,...,\alpha_m\in H(\mathbf{v})$. Since $f(\mathbf{0},\mathbf{0})=g(\mathbf{0})\neq 0,$ it is easy to prove that $\mathbf{d}_{\alpha_1},...,\mathbf{d}_{\alpha_m}$ are linearly independent. By {\bf Proposition \ref{4}}, $\mathbf{c}(u,\mathbf{v})$ is minimal.

{\bf Case 3:}  When $u\neq 0$ and $\mathbf{v}\neq \mathbf{0}$, let $\mathbf{\omega}=-\frac{1}{u}\mathbf{v}$ and $\mathbf{\omega}=({\omega_1},{\omega_2})$, where ${\omega_1}\in \mathbb{F}_q^s$, ${\omega_2}\in \mathbb{F}_q^t$.

{\bf If $\omega_1\neq \mathbf{0}$}, since $g(\beta)\equiv c\neq 0$, by {\bf Lemma \ref{311}}, there exist $\beta_1,...,\beta_{s}\in \mathbb{F}_q^s$ which are linearly independent, such that $\omega_1\cdot \beta_i =c,\ 1\leq i\leq s$.  Since $q>2$ and $\phi$ is an injection from $U$ to $\mathbb{F}_q^t\backslash \{\mathbf{0}\}$, there exists $a\in \mathbb{F}_q$ such that $\phi(a\mathbf{e}_1)\neq \omega_2$ and $\omega_1\cdot a\mathbf{e}_1\neq c$. Then by {\bf Lemma \ref{311}}, there exist $\gamma_1,...,\gamma_{t}\in \mathbb{F}_q^t$ which are linearly independent, such that $(\phi(a\mathbf{e}_1)-\omega_2)\cdot \gamma_j=\omega_1\cdot (a\mathbf{e}_1)-c,\ 1\leq j\leq t.$ Let
  $$
   \alpha_k=\left\{
            \begin{array}{ll}
             (\beta_k,\mathbf{0}), & \hbox{$1\leq k\leq s$},\\
             (a\mathbf{e}_1,\gamma_{k-s}), & \hbox{$s<k\leq m$}.
           \end{array}
          \right.
   $$

Then $\alpha_1,...,\alpha_m$ are linearly independent and $f(\alpha_k)=\omega\cdot \alpha_k,\ 1\leq k\leq m$. By {\bf Proposition \ref{3}}, $\mathbf{c}(u,\mathbf{v})$ is minimal.

{\bf If $\omega_1=\mathbf{0}\ \rm{and}\ \phi(\mathbf{0})\neq \omega_2$},  by {\bf Lemma \ref{311}}, there exist $\gamma_1,...,\gamma_{t}\in \mathbb{F}_q^t$ which are linearly independent, such that $(\phi(\mathbf{0})-\omega_2)\cdot \gamma_j=-c, \ 1\leq j\leq t.$ Since $q> 2$ and $\phi$ is an injection from $U$ to $\mathbb{F}_q^t\backslash \{\mathbf{0}\}$, there exist $a_i\in \mathbb{F}_q^*$ such that $\phi(a_i\mathbf{e}_i)\neq \omega_2$, $1\leq i\leq s$. By {\bf Lemma \ref{311}}, there exist $\gamma_{i0}\in \mathbb{F}_q^t$ such that $(\phi(a_i\mathbf{e}_i)-\omega_2)\cdot \gamma_{i0}=-c, \ 1\leq i\leq s.$  Let
  $$
   \alpha_k=\left\{
            \begin{array}{ll}
             (\mathbf{0},\gamma_k), & \hbox{$1\leq k\leq t$},\\
             (a_{k-t}\mathbf{e}_{k-t},\gamma_{(k-t)0}), & \hbox{$t<k\leq m$}.
           \end{array}
          \right.
   $$
Then $\alpha_1,...,\alpha_m$ are linearly independent and $f(\alpha_k)=\omega\cdot \alpha_k,\ 1\leq k\leq m$. By {\bf Proposition \ref{3}}, $\mathbf{c}(u,\mathbf{v})$ is minimal.

{\bf If $\omega_1=\mathbf{0}\ \rm{and}\ \phi(\mathbf{0})= \omega_2$}, since $\phi$ is an injection from $U$ to $\mathbb{F}_q^t\backslash \{\mathbf{0}\}$, $\phi(a\mathbf{e}_i)\neq \omega_2$ for any $a\in \mathbb{F}_q^*, \ 1\leq i\leq s.$ By {\bf Lemma \ref{311}}, there exist $\gamma_{11},...,\gamma_{1t}\in \mathbb{F}_q^t$ which are linearly independent, such that $(\phi(\mathbf{e}_1)-\omega_2)\cdot \gamma_{1j}=-c, \ 1\leq j\leq t,$ and there exist $\gamma_i\in \mathbb{F}_q^t$, such that $(\phi(\mathbf{e}_i)-\omega_2)\cdot \gamma_{i}=-c, \ 2\leq i\leq s.$   Let
  $$
   \alpha_k=\left\{
            \begin{array}{ll}
             (\mathbf{e}_1,\gamma_{1k}), & \hbox{$1\leq k\leq t$},\\
             (\mathbf{e}_{{k-t+1}},\gamma_{k-t+1}), & \hbox{$t<k\leq m-1$}.
           \end{array}
          \right.
   $$
Then $\alpha_1,...,\alpha_{m-1}$ are linearly independent and $f(\alpha_k)=\omega\cdot \alpha_k,\ 1\leq k\leq m-1$. Next, we will choose a suitable $\alpha_m.$

If there exists $a\in \mathbb{F}_q^*$, such that $\phi(a\mathbf{e}_1)-\omega_2 \notin \mathbb{F}_q(\phi(\mathbf{e}_1)-\omega_2)$, then there exists $\gamma_0\in \mathbb{F}_q^t$, such that
 \begin{equation} \label{eq}
\left\{ \begin{aligned}
(\phi(a\mathbf{e}_1)-\omega_2)\cdot \gamma_0& = -c\\
(\phi(\mathbf{e}_1)-\omega_2)\cdot \gamma_0 & \neq -ac.
\end{aligned}
\right.
\end{equation}
In this case, let $\alpha_m=(a\mathbf{e}_1,\gamma_0)$.

If for all $a\in \mathbb{F}_q^*$, $\phi(a\mathbf{e}_1)-\omega_2 \in \mathbb{F}_q(\phi(\mathbf{e}_1)-\omega_2)$, since $\phi$ is an injection, $\phi(\mathbf{e}_2)-\omega_2 \notin \mathbb{F}_q(\phi(\mathbf{e}_1)-\omega_2)$, and there exists $\eta\in \mathbb{F}_q^t$, such that
 \begin{equation} \label{eq}
\left\{ \begin{aligned}
(\phi(\mathbf{e}_2)-\omega_2)\cdot \eta & = 0\\
(\phi(\mathbf{e}_1)-\omega_2)\cdot \eta & =1.
\end{aligned}
\right.
\end{equation}
In this case, let $\alpha_m=(\mathbf{e}_2,\gamma_2+\eta)$.

We can verify $\alpha_1,...,\alpha_m$ are linearly independent and $f(\alpha_k)=\omega\cdot \alpha_k,\ 1\leq k\leq m$. By {\bf Proposition \ref{3}}, $\mathbf{c}(u,\mathbf{v})$ is minimal.

This completes the proof.

 \end{proof}

 \begin{thm}\label{42}
Let $q= 2$, $s\geq 2$ and $t\geq 2$. Let $U=\{\beta\in \mathbb{F}_q^s: \mathbf{wt}(\beta)\leq 1\}.$ If $\phi$ and $g$ satisfy:\\
$(1)$ $\phi$ is an injection from $U$ to $\mathbb{F}_q^t\backslash \{\mathbf{0}\}$, and for all $\beta$ satisfying $\mathbf{wt}(\beta)=2$, $\phi(\beta)=\mathbf{0}$,\\
$(2)$ $g(\beta)\equiv 1$,\\
then $\mathcal{C}_f$ is a $[q^m-1, m+1]_q$ minimal linear code.

\end{thm}
\begin{proof} By condition (2), it is easy to see $f(\mathbf{x})\neq\mathbf{\omega\cdot x}$ for any $\mathbf{\omega}\in \mathbb{F}_q^m$, then by {\bf Corollary \ref{332}}, $\mathcal{C}_f$ is a $[q^m-1,m+1]_q$ linear code.
Next, we will prove the codewords $\mathbf{c}(u,\mathbf{v})$ are minimal in three different cases.

{\bf Case 1:}  When $u\neq 0$ and $\mathbf{v}=\mathbf{0}$,  the proof is the same as that in {\bf Theorem \ref{41}}.

{\bf Case 2:}  When $u\neq 0$ and $\mathbf{v}=\mathbf{0}$, the proof is the same as that in {\bf Theorem \ref{41}}.

{\bf Case 3:} When $u\neq 0$ and $\mathbf{v}\neq \mathbf{0}$, let $\mathbf{\omega}=-\frac{1}{u}\mathbf{v}$ and $\mathbf{\omega}=({\omega_1},{\omega_2})$, where ${\omega_1}\in \mathbb{F}_q^s$, ${\omega_2}\in \mathbb{F}_q^t$.

{\bf If $\omega_1\neq \mathbf{0}$}, since $g(\mathbf{x})\equiv 1\neq 0$, by {\bf Lemma \ref{311}}, there exist $\beta_1,...,\beta_{s}\in \mathbb{F}_q^s$ which are linearly independent, such that $\omega_1\cdot \beta_i =1.$

(1) When $\phi(\mathbf{0})\neq \omega_2,$ by {\bf Lemma \ref{311}}, there exist $\gamma_1,...,\gamma_{t}\in \mathbb{F}_q^t$ which are linearly independent, such that $(\phi(\mathbf{0})-\omega_2)\cdot \gamma_j=-1=\omega_1\cdot \mathbf{0}-g(\mathbf{0}),\ 1\leq j\leq t.$
 Let
  $$
   \alpha_k=\left\{
            \begin{array}{ll}
             (\beta_k,\mathbf{0}), & \hbox{$1\leq k\leq s$},\\
             (\mathbf{0},\gamma_{k-s}), & \hbox{$s<k\leq m$}.
           \end{array}
          \right.
   $$

 (2) When $\phi(\mathbf{0})= \omega_2$ and $\omega_1\cdot \mathbf{e}_1-1\neq 0$, since $\phi$ is an injection from $U$ to $\mathbb{F}_q^t\backslash \{\mathbf{0}\}$, $\phi(\mathbf{e}_1)\neq \omega_2$. Then by {\bf Lemma \ref{311}}, there exist $\gamma_1,...,\gamma_{t}\in \mathbb{F}_q^t$ which are linearly independent, such that $(\phi(\mathbf{e}_1)-\omega_2)\cdot \gamma_j=\omega_1\cdot \mathbf{e}_1-1=\omega_1\cdot \mathbf{e}_1-g(\mathbf{e}_1),\ 1\leq j\leq t.$
 Let
  $$
   \alpha_k=\left\{
            \begin{array}{ll}
             (\beta_k,\mathbf{0}), & \hbox{$1\leq k\leq s$},\\
             (\mathbf{e}_1,\gamma_{k-s}), & \hbox{$s<k\leq m$}.
           \end{array}
          \right.
   $$

 (3) When $\phi(\mathbf{0})= \omega_2$ and $\omega_1\cdot \mathbf{e}_1-1= 0$, since $\phi$ is an injection from $U$ to $\mathbb{F}_q^t\backslash \{\mathbf{0}\}$, $\phi(\mathbf{e}_i)\neq \omega_2$ for $1\leq i\leq s$. Then by {\bf Lemma \ref{311}}, there exist $\gamma_1,...,\gamma_{t-1}\in \mathbb{F}_q^t$ which are linearly independent, such that $(\phi(\mathbf{e}_1)-\omega_2)\cdot \gamma_j=\omega_1\cdot \mathbf{e}_1-1=\omega_1\cdot \mathbf{e}_1-g(\mathbf{e}_1),\ 1\leq j\leq t-1.$ Since $\phi(\mathbf{e}_2)\neq \phi(\mathbf{e}_1),$ $\phi(\mathbf{e}_2)-\omega_2$ and $\phi(\mathbf{e}_1)-\omega_2$ are linearly independent. By {\bf Lemma \ref{311}}, there exists $\gamma'$ such that

 \begin{equation} \label{eq}
\left\{ \begin{aligned}
(\phi(\mathbf{e}_1)-\omega_2)\cdot \gamma'& = 1 \\
(\phi(\mathbf{e}_2)-\omega_2)\cdot \gamma' & = \omega_1\cdot \mathbf{e}_2-1.
\end{aligned}
\right.
\end{equation}
 Let
  $$
   \alpha_k=\left\{
            \begin{array}{ll}
             (\beta_k,\mathbf{0}), & \hbox{$1\leq k\leq s$},\\
             (\mathbf{e}_1,\gamma_{k-s}), & \hbox{$s<k\leq m-1$},\\
             (\mathbf{e}_2,\gamma'), & \hbox{$k=m$}.
           \end{array}
          \right.
   $$

In all three cases (1),(2) and (3), $\alpha_1,...,\alpha_m$ are linearly independent and $f(\alpha_k)=\omega\cdot \alpha_k,\ 1\leq k\leq m$. By {\bf Proposition \ref{3}}, $\mathbf{c}(u,\mathbf{v})$ is minimal.

{\bf If $\omega_1=\mathbf{0}\ \rm{and}\ \phi(\mathbf{0})\neq \omega_2$},  by {\bf Lemma \ref{311}}, there exist $\gamma_1,...,\gamma_{t}\in \mathbb{F}_q^t$ which are linearly independent, such that $(\phi(\mathbf{0})-\omega_2)\cdot \gamma_j=1, \ 1\leq j\leq t.$ Since $s\geq 2$ and $\phi$ is an injection from $U$ to $\mathbb{F}_q^t\backslash \{\mathbf{0}\}$, there exists $\mathbf{e}_{i_0}$ such that $\phi(\mathbf{e}_{i_0})\neq \omega_2$. For $1\leq i\leq s,$ let
$$\beta_i=\left\{
            \begin{array}{ll}
             \mathbf{e}_{i_0}+\mathbf{e}_{i}, & \hbox{$i\neq i_0$},\\
             \mathbf{e}_{i_0}. & \hbox{$i=i_0$},
           \end{array}
          \right.$$
Then $\phi(\beta_i)\neq \omega_2.$
By {\bf Lemma \ref{311}}, there exist $\gamma_{i1}\in \mathbb{F}_q^t$ such that $(\phi(\beta_i)-\omega_2)\cdot \gamma_{i1}=1, \ 1\leq i\leq s.$  Let
  $$
   \alpha_k=\left\{
            \begin{array}{ll}
             (\mathbf{0},\gamma_k), & \hbox{$1\leq k\leq t$},\\
             (\beta_{k-t},\gamma_{(k-t)1}), & \hbox{$t<k\leq m$}.
           \end{array}
          \right.
   $$
Then $\alpha_1,...,\alpha_m$ are linearly independent and $f(\alpha_k)=\omega\cdot \alpha_k,\ 1\leq k\leq m$. By {\bf Proposition \ref{3}}, $\mathbf{c}(u,\mathbf{v})$ is minimal.

{\bf If $\omega_1=\mathbf{0}\ \rm{and}\ \phi(\mathbf{0})= \omega_2$}, since $\phi$ is an injection from $U$ to $\mathbb{F}_q^t\backslash \{\mathbf{0}\}$, $\phi(\mathbf{e}_i)\neq \omega_2$, $1\leq i\leq s.$ By {\bf Lemma \ref{311}}, there exist $\gamma_{11},...,\gamma_{1t}\in \mathbb{F}_q^t$ which are linearly independent, such that $(\phi(\mathbf{e}_1)-\omega_2)\cdot \gamma_{1j}=1, \ 1\leq j\leq t,$ and there exists $\gamma_i\in \mathbb{F}_q^t$, such that $(\phi(\mathbf{e}_i)-\omega_2)\cdot \gamma_{i}=1, \ 2\leq i\leq s.$   Let
  $$
   \alpha_k=\left\{
            \begin{array}{ll}
             (\mathbf{e}_1,\gamma_{1k}), & \hbox{$1\leq k\leq t$},\\
             (\mathbf{e}_{k-t+1},\gamma_{k-t+1}), & \hbox{$t<k\leq m-1$}.
           \end{array}
          \right.
   $$
Then $\alpha_1,...,\alpha_{m-1}$ are linearly independent and $f(\alpha_k)=\omega\cdot \alpha_k,\ 1\leq k\leq m-1$. Next, we will choose a suitable $\alpha_m.$

Since $\phi(\mathbf{e}_2)\neq \phi(\mathbf{e}_1),$ $\phi(\mathbf{e}_2)-\omega_2$ and $\phi(\mathbf{e}_1)-\omega_2$ are linearly independent. By {\bf Lemma \ref{311}}, there exists $\eta\in \mathbb{F}_q^t$ such that
 \begin{equation} \label{eq}
\left\{ \begin{aligned}
(\phi(\mathbf{e}_2)-\omega_2)\cdot \eta & = 0\\
(\phi(\mathbf{e}_1)-\omega_2)\cdot \eta & =1.
\end{aligned}
\right.
\end{equation}
Let $\alpha_m=(\mathbf{e}_2,\gamma_2+\eta)$.

We can verify $\alpha_1,...,\alpha_m$ are linearly independent and $f(\alpha_k)=\omega\cdot \alpha_k,\ 1\leq k\leq m$. By {\bf Proposition \ref{3}}, $\mathbf{c}(u,\mathbf{v})$ is minimal.

\end{proof}

{\bf Theorem \ref{41}} and {\bf Theorem \ref{42}} generalize the following result in \cite{DHZ2018} by Ding  et al..

 \begin{prop}\label{313}\cite[Theorem 23]{DHZ2018}
 Let $m\geq 7$ be an odd integer, $s:=(m+1)/2$, and $t:=(m-1)/2$. Let $U=\{\beta\in \mathbb{F}_q^s: \mathbf{wt}(\beta)\leq 1\}.$ and $V=\{\mathbf{0}\}$. Let $f$ be the Boolean function defined as $f(\mathbf{x},\mathbf{y})=\phi(\mathbf{x})\cdot \mathbf{y}+g(\mathbf{x})$, where $g\equiv 1$, and $\phi$ is an injection from $U$ to $GF(2)^t\setminus V$ and $\phi(\mathbf{x})=\mathbf{0}$ for any $\mathbf{\mathbf{x}}\in GF(2)^s\setminus U.$ Then the code $\mathcal{C}_f$ is a $[2^m-1,m+1,2^{m-1}-2^{t-1}(s-1)]$ binary minimal code with $\omega_{min}/\omega_{max}\leq 1/2.$
 \end{prop}
\begin{remark}
 Our theorems generalize {\bf Proposition \ref{313}}  from the following three aspects. First, we generalize $q=2$ in \cite{DHZ2018} to all prime powers. Second, in our construction, the choice of  integers $m,s,t$ are more flexible  than theirs. Third, in our construction, the choice of function $\phi$ is also more flexible: in \cite{DHZ2018}, $\phi(\mathbf{x})$ should be $\mathbf{0}$ when ${\bf wt(x)}\geq 2$, while in our constructions, if $q>2$, $\phi(\mathbf{x})$ has no any restrictions when ${\bf wt(x)}\geq 2$, and if $q=2$, $\phi(\mathbf{x})$ has no any restrictions when ${\bf wt(x)}\geq 3$. Then the functions $f$ in our constructions  can be more flexible.
\end{remark}
\begin{remark}
In \cite[Theorem 3.1,\ 4.1]{XQC2020}, the authors also generalize {\bf Proposition \ref{313}}. Our result is different with theirs. First,  In \cite{XQC2020} the linear codes are over $\mathbb{F}_p$, where $p$ is a prime, while linear codes in our construction are over $\mathbb{F}_q$, where $q$ is a prime power. Second, in \cite{XQC2020}, $g(\mathbf{x})\equiv 0$, while in {\bf Theorem \ref{41}} and {\bf Theorem  \ref{42}} $g(\mathbf{x})\equiv c\neq 0.$ Third, the set $U$ is different.

\end{remark}

\begin{exm}
Let $m=7$, $s=4$,  $t=3$, $\mathbf{x}=(x_1,x_2,x_3,x_4)$ and $\phi$ be a function from $\mathbb{F}_q^4\rightarrow\mathbb{F}_q^3$ satisfying
$$\phi(\mathbf{x})=\left\{
            \begin{array}{ll}
            (x_1,x_2,x_3), & \hbox{${\bf wt}({\bf x})\leq 1$ and $(x_1,x_2,x_3)\neq(0,0,0)$},\\
            (x_4,x_4,x_4), & \hbox{${\bf wt}({\bf x})=1$ and $(x_1,x_2,x_3)=(0,0,0)$},\\
            (1,0,0), & \hbox{otherwise}.
           \end{array}
          \right.$$
Let $g$ is a function from $\mathbb{F}_q^4\rightarrow\mathbb{F}_q$ satisfying $g({\mathbf{x}})\equiv 1.$

(1) When $q=2$,  by {\bf Theorem \ref{42}}, $\mathcal{C}_{f}$ is a minimal linear code over $\mathbb{F}_2$ with parameters $[127,\ 8,\ 39]$ and weight enumerator
 $$1+z^{39}+12z^{55}+8z^{59}+72z^{63}+127z^{64}+24z^{67}+10z^{71}+z^{103}.$$
  Furthermore, $\frac{w_{\rm min}}{w_{\rm max}}=\frac{39}{103}<\frac{1}{2}.$ The parameter of Example 24 in \cite{DHZ2018} is $[127,\ 8,\ 28]$, so $\mathcal{C}_{f}$ in our construction is a new minimal linear code, and the minimum distance is better than that in \cite{DHZ2018}.

(2) When $q=3$,  by {\bf Theorem \ref{41}}, $\mathcal{C}_{f}$ is a minimal linear code over $\mathbb{F}_3$ with parameters $[2186,\ 8,\ 1295]$ and weight enumerator
 $$1+2z^{1295}+18z^{1376}+90z^{1403}+108z^{1439}+3588z^{1457}+2186z^{1458}+378z^{1466}+180z^{1484}+8z^{1538}+2z^{2024}.$$
  Furthermore, $\frac{w_{\rm min}}{w_{\rm max}}=\frac{1295}{2024}<\frac{2}{3}.$ The parameter of Example 2 in \cite{XQC2020} is $[2186,\ 8,\ 162]$, so $\mathcal{C}_{f}$ in our construction is a new minimal linear code, and the minimum distance is better than that in \cite{XQC2020}.

\end{exm}

\section{\bf{Minimal linear codes constructed from a class of polynomials}}\label{polynomials}
In this section, we present a construction of minimal linear codes from polynomials which generalize the result in \cite{MB2019}.
Let $m$ be a positive integer. We consider the polynomials $f$ from  $\mathbb{F}_q^m$ $\rightarrow$ $\mathbb{F}_q$. Let $g(\mathbf{x})=\prod_{i=1}^m x_i^{b_i}$ be a monomial. Define
$$s(g):=\{i\mid b_i\neq 0, 1\leq i\leq m\}.$$

Let  $f(\mathbf{x})=\sum_{j=1}^t a_jg_j(\mathbf{x})$ be a polynomial from $\mathbb{F}_q^m\rightarrow \mathbb{F}_q$, where $a_j\in \mathbb{F}_q^*$ and $g_j(\mathbf{x})$ are monomials. Let
$$\mathcal{C}_f:=\{\mathbf{c}(u,\mathbf{v}):=((u,\mathbf{v})\cdot\mathbf{d_x})_{\mathbf{x}\in \mathbb{F}_q^m\backslash \{\mathbf{0}\}}=(uf(\mathbf{x})+\mathbf{v\cdot x})_{\mathbf{x}\in \mathbb{F}_q^m\backslash \{\mathbf{0}\}}\mid\ u\in \mathbb{F}_q,\mathbf{v}\in \mathbb{F}_q^m\}.$$

We will give some sufficient conditions for $\mathcal{C}_f$ to be minimal.

\begin{thm}\label{t1}
Let $t\geq 2$, $f(\mathbf{x})=\sum_{j=1}^t a_jg_j(\mathbf{x})$, where $a_j\in \mathbb{F}_q^*$ and $g_j(\mathbf{x})$ are monomials. If $\{g_j(\mathbf{x}):\ 1\leq j\leq t\}$ satisfy the following two conditions:\\
$(1)$ $\{s(g_j):\ 1\leq j\leq t\}$ are disjoint,\\
$(2)$ $\#s(g_j)\geq 3$ for all $1\leq j\leq t$,\\
then $\mathcal{C}_f$ is a $[q^m-1, m+1]_q$ minimal linear code.

\end{thm}

\begin{proof}
By condition (2), it is easy to see that $f(\mathbf{x})\neq\mathbf{\omega\cdot x}$ for any $\mathbf{\omega}\in \mathbb{F}_q^m$, then by {\bf Corollary \ref{332}}, $\mathcal{C}_f$ is a $[q^m-1,m+1]_q$ linear code.
Next, we will prove that the codewords $\mathbf{c}(u,\mathbf{v})$ are minimal in three different cases.

By condition (2), we know that if $1\leq \mathbf{wt(x)}\leq 2$, $f(\mathbf{x})=0$.

 {\bf Case 1}: When $u\neq 0$ and $\mathbf{v}=\mathbf{0}$, since $\mathbf{wt(e}_i)=1$, $f(\mathbf{e}_i)=0,\ 1\leq i\leq m$. By {\bf Proposition \ref{1}}, $\mathbf{c}(u,\mathbf{v})$ is minimal.

 {\bf Case 2}: When $u=0$ and $\mathbf{v}\neq\mathbf{0}$, let $\mathbf{v}=(v_1,...,v_m)$. Then there exists $v_{i_0}\neq 0$. For $i\neq i_0$, let
 $$\alpha_i={\bf e}_{i}-v_i v_{i_0}^{-1}{\bf e}_{i_0}.$$
  Then $\{\alpha_i: \ i\neq i_0\}$ is a basis of $H(\mathbf{v})$ and $f(\alpha_i)=0$ for $i\neq i_0$. Next, we will choose a suitable $\alpha_{i_0}.$

  Since $t\geq 2$ and $\{s(g_j):\ 1\leq j\leq t\}$ are disjoint, there exists $j_1$ such that $i_0\not\in s(g_{j_1})$. Let $\alpha_{i_0}:=\Sigma_{i\in s(g_{j_1})} \alpha_i.$ Then  $\alpha_{i_0}\in H(\mathbf{v})$ and $f(\alpha_{i_0})=a_{j_1}\neq 0$. Since $\Sigma_{i\in s(g_{j_1})} f(\alpha_i)=0\neq f(\alpha_{i_0})$, it is easy to see that $\mathbf{d}_{\alpha_1},...,\mathbf{d}_{\alpha_m}$ are linearly independent.
  By {\bf Proposition \ref{4}}, $\mathbf{c}(u,\mathbf{v})$ is minimal.

{\bf Case 3}: When $u\neq 0$ and $\mathbf{v}\neq\mathbf{0}$, let $\mathbf{\omega}=(w_1,...,w_m)=-\frac{1}{u}\mathbf{v}$. Then there exists $w_{i_0}\neq 0.$ For $i\neq i_0$, let $\alpha_i=\mathbf{e}_i-w_iw_{i_0}^{-1}\mathbf{e}_{i_0}.$ It is easy to see that $\{\alpha_i:\ i\neq i_0\}$ constitutes a basis of $H(\omega)$ and $\omega\cdot \alpha_i=0$. Since $1\leq \mathbf{wt}(\alpha_i)\leq 2$,  $f(\alpha_i)=0=\omega\cdot \alpha_i$, $i\neq i_0$. Next, we will choose a suitable $\alpha_{i_0}.$

 Since $t\geq 2$ and $\{s(g_j):\ 1\leq j\leq t\}$ are disjoint, there exists $j_1$ such that $i_0\not\in s(g_{j_1})$. Let $\alpha_{i_0}:=\Sigma_{i\in s(g_{j_1})} \alpha_i+a_{j_1}w_{i_0}^{-1}e_{i_0}.$ Then $\omega\cdot \alpha_{i_0}=a_{j_1}\neq 0$ and $f(\alpha_{i_0})=a_{j_1}$.

Thus, $\{\alpha_1,...,\alpha_m\}$ is a basis of $\mathbb{F}_q^m$, and $f(\alpha_{i})=\omega\cdot \alpha_{i}$. By {\bf Proposition \ref{3}},  $\mathbf{c}(u,\mathbf{v})$ is minimal.

This completes the proof.
\end{proof}

\begin{thm}\label{t2}
Let $t\geq 2$, $f(\mathbf{x})=\sum_{j=1}^t a_jg_j(\mathbf{x})$, where $a_j\in \mathbb{F}_q^*$ and $g_j(\mathbf{x})$ are monomials. If $\{g_j(\mathbf{x}):\ 1\leq j\leq t\}$ satisfy the following three conditions:\\
$(1)$ $\{s(g_j):\ 1\leq j\leq t\}$ are disjoint,\\
$(2)$ $\#s(g_j)\geq 2$ for all $1\leq j\leq t$,\\
$(3)$ for $g_j(\mathbf{x})=\Pi_{i=1}^m x_i^{b_i}$, $b_i\in \{0,1\}, 1\leq i\leq m$, \\
then $\mathcal{C}_f$ is a $[q^m-1,m+1]_q$ minimal linear code.

\end{thm}

\begin{proof}
By condition (2), it is easy to see $f(\mathbf{x})\neq\mathbf{\omega\cdot x}$ for any $\mathbf{\omega}\in \mathbb{F}_q^m$, then by {\bf Corollary \ref{332}}, $\mathcal{C}_f$ is a $[q^m-1,m+1]_q$ linear code. Next, we will prove the codewords $\mathbf{c}(u,\mathbf{v})$ are minimal in three different cases.

By condition (2), we know that if $\mathbf{wt(x)}=1$, $f(\mathbf{x})=0$.

 {\bf Case 1}:  When $u\neq 0$ and $\mathbf{v}=\mathbf{0}$, the proof is the same as that in {\bf Theorem \ref{t1}}.

 {\bf Case 2}:   When $u=0$ and $\mathbf{v}\neq \mathbf{0}$, the proof is also the same as that in {\bf Theorem \ref{t1}}.

{\bf Case 3}: When $u\neq 0$ and $\mathbf{v}\neq\mathbf{0}$, let $\mathbf{\omega}=(w_1,...,w_m)=-\frac{1}{u}\mathbf{v}$. Then there exists $w_{i_0}\neq 0.$ For $i\neq i_0$, let $\alpha_i=\mathbf{e}_i-w_iw_{i_0}^{-1}\mathbf{e}_{i_0}.$ It is easy to see $\{\alpha_i:\ i\neq i_0\}$ constitutes a basis of $H(\omega)$ and $\omega\cdot \alpha_i=0$.  Since $t\geq 2$ and $\{s(g_j):\ 1\leq j\leq t\}$ are disjoint, there exists $j_1$ such that $i_0\not\in s(g_{j_1})$. Let  $\alpha_{i_0}:=\sum_{i\in s(g_{j_1})} \alpha_i+a_{j_1}w_{i_0}^{-1}e_{i_0}.$ Then $f(\alpha_{i_0})=a_{j_1}=\omega\cdot \alpha_{i_0},\ a_{j_1}\neq 0$ and $\{\alpha_1,...,\alpha_m\}$ is a basis of $\mathbb{F}_q^m$.

{\bf If} for all $i\neq i_0$, $f(\alpha_{i})=0=\omega\cdot \alpha_{i}$, thus for all $1\leq i\leq m,$ $f(\alpha_{i})=\omega\cdot \alpha_{i}$. Then by {\bf Proposition \ref{3}},  $\mathbf{c}(u,\mathbf{v})$ is minimal.

{\bf If } there exists  $i_1\neq i_0$, such that $f(\alpha_{i_1})\neq 0$, then $w_{i_1}\neq 0$, and there exists $j_0$ such  that $s(g_{j_0})=\{i_0,i_1\}$. Since $\{s(g_j):\ 1\leq j\leq t\}$ are disjoint,   $f(\alpha_{i})=0=\omega\cdot \alpha_{i}$ for all $i\neq i_0,i_1$. But $f(\alpha_{i_1})\neq \omega\cdot \alpha_{i_1}$, so we need to choose a suitable $\beta \in \mathbb{F}_q^m$ satisfying $f(\beta)=\omega \cdot \beta$. We have the following two cases:

(a) If there exists $i_2\in  s(g_{j_1})$, such that $w_{i_2}\neq 0$, let $\beta=\alpha_{i_1}-w_{i_2}^{-1}w_{i_1}\alpha_{i_2}.$

(b) If $w_i=0$ for all $i\in s(g_{j_1})$, choose any $i_2\in s(g_{j_1})$ and
$$
 k_i:=\left\{
            \begin{array}{ll}
             1, & \hbox{$i\in s(g_{j_1})\ and\ \ i\neq i_2$},\\
              a_{j_1}^{-1}a_{j_0}w_{i_1}w_{i_0}^{-1}, & \hbox{$i=i_2$.}
           \end{array}
          \right.
$$
Let $\beta=\alpha_{i_1}+\sum_{i\in s(g_{j_1})}k_i\alpha_i.$

Both in cases (a) and (b), $\beta$ satisfies $f(\beta)=0=\beta\cdot \omega$. Let
$$
 \beta_i:=\left\{
            \begin{array}{ll}
             \alpha_i, & \hbox{$i\neq i_1$},\\
             \beta, & \hbox{$i=i_1$.}
           \end{array}
          \right.
$$
Then $\{\beta_i:1\leq i\leq m\}$ constitutes a basis of $\mathbb{F}_q^m$ and $f(\beta_{i})=\omega\cdot \beta_{i}$ for $1\leq i\leq m.$
By {\bf Proposition \ref{3}},  $\mathbf{c}(u,\mathbf{v})$ is minimal.

This completes the proof.
\end{proof}

In {\bf Theorem \ref{t2}}, take $m=tr,\ a_j=1,$ and $g_j(\mathbf{x})=x_{jr+1}x_{jr+2}...x_{jr+r},\ 1\leq j\leq t$. Then we have the following corollary.

\begin{cor}\label{411}
Let $m=tr$ and $f(\mathbf{x})=\sum_{j=1}^t x_{jr+1}x_{jr+2}...x_{jr+r}.$ If $t\geq 2$ and $r\geq 2$, then $\mathcal{C}_f$ is minimal.
\end{cor}

\begin{remark}
This {\bf Corollary} is exactly the result in \cite[Corollary 5.2]{MB2019}, so {\bf Theorem \ref{t2}} contains  Corollary 5.2 in \cite{MB2019} as a special case.
\end{remark}

\begin{remark}
As in {\bf Theorem \ref{t1}}, when $\#s(g_j)\geq 3$ for all $1\leq j\leq t$, no more restrictions are needed on $g_j(\mathbf{x})$, while in {\bf Theorem \ref{t2}}, when $\#s(g_j)\geq 2$ for all $1\leq j\leq t$, $g_j(\mathbf{x})$ should be square-free. In our construction, the choice of functions $f$ is more flexible than in \cite{MB2019}.
\end{remark}

\begin{exm}
Let $q=3$, $m=8$. Then $n=q^m-1=6560$ and $k=m+1=9$. With the help of Magma, we get the following results.

(1) Let $f_1(\mathbf{x})=x_1x_2x_3x_4+x_5x_6x_7x_8$. Then  $\mathcal{C}_{f_1}$ is a minimal linear code over $\mathbb{F}_3$ with parameters $[6560,\ 9,\ 2208]$ and $w_{\rm max}=4602$. Furthermore, $\frac{w_{\rm min}}{w_{\rm max}}=\frac{2208}{4602}<\frac{2}{3}.$

(2)  Let $f_2(\mathbf{x})=x_1x_2+x_3x_4+x_5x_6+x_7x_8$. Then  $\mathcal{C}_{f_2}$ is a minimal linear code over $\mathbb{F}_3$ with parameters $[6560,\ 9,\ 4320]$ and $w_{\rm max}=4401$. Furthermore, $\frac{w_{\rm min}}{w_{\rm max}}=\frac{4302}{4401}>\frac{2}{3}.$

(3)  Let $f_3(\mathbf{x})=x_1x_2x_3+x_4x_5x_6x_7x_8$. Then  $\mathcal{C}_{f_3}$ is a minimal linear code over $\mathbb{F}_3$ with parameters $[6560,\ 9,\ 2424]$ and $w_{\rm max}=4764$. Furthermore, $\frac{w_{\rm min}}{w_{\rm max}}=\frac{2424}{4764}<\frac{2}{3}.$

(4)  Let $f_4(\mathbf{x})=x_1x_2x_3+x_4x_5x_6x_7$. Then  $\mathcal{C}_{f_4}$ is a minimal linear code over $\mathbb{F}_3$ with parameters $[6560,\ 9,\ 2664]$ and $w_{\rm max}=4716$. Furthermore, $\frac{w_{\rm min}}{w_{\rm max}}=\frac{2664}{4716}<\frac{2}{3}.$

When $m=8$, $f_1$ and $f_2$ are all the cases provided by \cite[Corollary 5.2]{MB2019} , $f_3$ and $f_4$ are the cases constructed in {\bf Theorem \ref{t2}}. we can see $\mathcal{C}_{f_3}$ and $\mathcal{C}_{f_4}$ are not equivalent to $\mathcal{C}_{f_1}$ and $\mathcal{C}_{f_2}$ respectively, since the minimum distance are different. In fact, the choices of $f$ in {\bf Theorem \ref{t2}} are very flexible.

\end{exm}

\section {\bf{Concluding remarks}}\label{section Concluding remarks}
In this paper, we first give the sufficient and necessary condition for
linear codes constructed from $q$-ary functions to be minimal.  Based on this sufficient and necessary condition, we give four constructions of minimal linear codes and generalize the results in \cite{DHZ2018, HDZ2018, BB2019, MQ2019, MB2019}. The choices of $f$ are much more flexible than theirs. Because of our new sufficient and necessary condition, it is easy to prove that the linear codes we presented are minimal. We expect that based on our sufficient and necessary condition, more minimal linear codes can be constructed from the functions.

 {}
\end{document}